\documentclass{article}
\usepackage{spconf}
\usepackage{amsmath,amssymb,amsfonts}
\usepackage{algorithmic}
\usepackage{graphicx}
\usepackage{textcomp}
\usepackage{xcolor}
\renewcommand{\vec}[1]{\mathbf{#1}}
\usepackage{algorithm}
\usepackage{algorithmic}
\usepackage{graphicx} 
\usepackage{float} 
\usepackage{stfloats}
\usepackage[section]{placeins}
\usepackage{subfigure}
\usepackage{booktabs}
\usepackage[numbers,sort&compress]{natbib}
\usepackage{amssymb}
\usepackage{setspace}
\usepackage{multirow}
\usepackage{color}
\usepackage{hyperref}
\usepackage{pifont}
\usepackage{bm}
\usepackage{amsthm}

\newtheorem{proposition}{Proposition}
\newtheorem{remark}{Remark}
\usepackage{epsfig}

\title{Optimal Mixed-ADC arrangement for DOA Estimation via CRB using ULA}
\name{Xinnan Zhang, Yuanbo Cheng, Xiaolei Shang and Jun Liu \thanks{This work was supported in part by the NSFC under award number 62271461, the Youth Innovation Promotion Association CAS (CX2100060053) and the Anhui Provincial Natural Science Foundation under Grant 2208085J17.}}
\address{Department of EEIS, University of Science and Technology of China, Hefei, Anhui, China\\
Emails: $\{$zhangxinnan,cyb967,xlshang$\}$@mail.ustc.edu.cn; junliu@ustc.edu.cn
}

\begin{document}
%
\maketitle
\begin{abstract}
We consider a mixed analog-to-digital converter (ADC) based architecture
for direction of arrival (DOA) estimation using a uniform linear array (ULA). We derive the Cram{\'e}r-Rao bound (CRB) of the DOA under the optimal time-varying threshold, and find that the asymptotic CRB is related to the arrangement of high-precision and one-bit ADCs for a fixed number of ADCs. Then, a new concept called ``mixed-precision arrangement" is proposed. It is proven that better performance for DOA estimation is achieved when high-precision ADCs are distributed evenly around the edges of the ULA. This result can be extended to a more general case where the ULA is equipped with various precision ADCs. Simulation results show the validity of the asymptotic CRB and better performance under the optimal mixed-precision arrangement.
\end{abstract}

\section{Introduction}

The problem of direction of arrival (DOA) estimation is of great importance in the field of array signal processing with many applications in automotive radar, sonar, wireless communications \cite{tuncer2009classical,sun2020mimo,krim1996two}. 

To ensure accuracy, high-precision analog-to-digital converters (ADCs) are often employed. However, the power consumption and hardware cost of ADC increase exponentially as the quantization bit  and sampling rate grow \cite{walden1999analog}. Using one-bit ADCs is a promising technique  \cite{1039405} to mitigate the aforementioned ADC problems. Recently, one-bit sampling based on time-varying threshold schemes has been considered in \cite{8822763, 8291043,9930674,eamaz2021modified}, which can eliminate ambiguity between the signal amplitude and noise variance. Nevertheless, the pure one-bit ADC system suffers from lots of problems like large rate loss especially in high signal-to-noise ratio (SNR) regime \cite{mo2014high} and dynamic range problem, i.e., a strong target can mask a weak target \cite{walden1999analog}.

Instead, a mixed-ADC based architecture has been proposed in \cite{liang2016mixed} to overcome the above weakness, where most receive antenna outputs are sampled by one-bit ADCs and a few outputs are sampled by high-resolution ADCs. Under the mixed-ADC architecture, some works have been proposed to analysis the DOA performance loss under uniform linear arrays (ULAs)  \cite{shi2022performance} and Cram{\'e}r-Rao bound (CRB) in phase-modulated continuous-wave multiple-input multiple-output radar \cite{shang2022mixed,10027928}. However, they used the mixed-ADC architecture by simply employing high-precision ADCs on one side and one-bit ADCs on the other side (see  Fig. 1(a) for an example), which wastes some potentialities of the mixed-ADC based architecture. 
\begin{figure}[htbp]\label{fig:arr}
\centering
\includegraphics[scale=0.14]{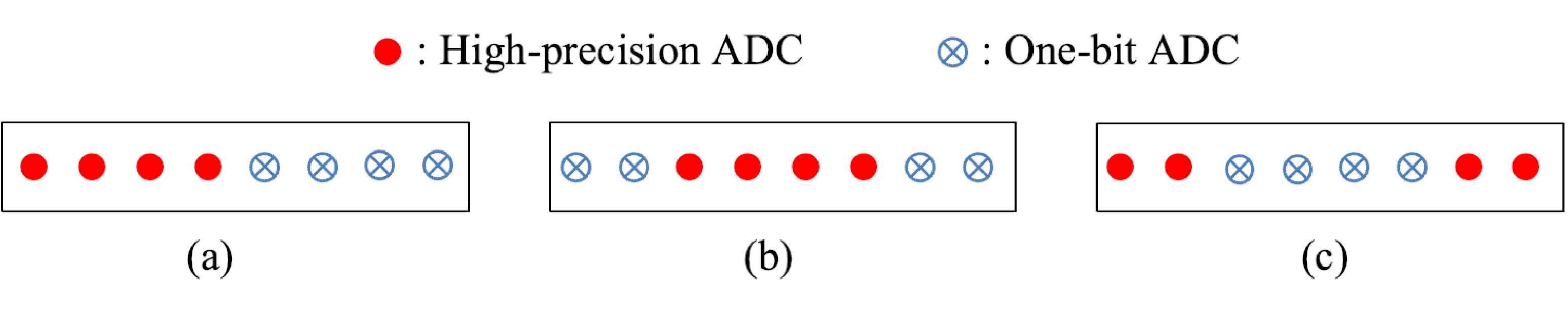}
\caption{Different mixed-ADC arrangements using ULA}
\label{figure}
\end{figure}

In this work, we derive the CRB associated with DOA under the mixed-ADC based architecture with time-varying threshold. For computational simplicity, we consider the asymptotic CRB and find it is related to the arrangement of high-precision and one-bit ADCs. Furthermore, we propose a new concept named ``mixed-precision arrangement''. It is found that the high-precision ADCs should be arranged evenly around the edges of the ULA to achieve a lower CRB like the case in Fig. 1(c). Numerical results demonstrate that the asymptotic CRB is valid and the optimal mixed-precision arrangement can achieve better performance.

\indent \textit{Notation:} We denote vectors and matrices by bold lowercase and uppercase letters, respectively.  $(\cdot)^T$ and $(\cdot)^H$ represent the  transpose and the conjugate transpose, respectively. $\vec{I}_N$ denotes an $N\times N$ identity matrix and $\mathbf{1}_N=[1, \ldots, 1]^T \in \mathbb{R}^{N \times 1}$.
$\otimes$, $\odot$ and $\circ$ denote the  Kronecker, Khatri–Rao and Hadamard matrix products, respectively. $\rm{vec}(\cdot)$  refers to the column-wise vectorization
operation and $\rm{diag}(\vec{d})$ denotes a diagonal matrix with diagonal
entries formed from $\vec{d}$. 
$\mathbf{A}_{\mathrm{R}}\triangleq \Re\{\mathbf{A}\}$ and
$\mathbf{A}_{\mathrm{I}} \triangleq \Im\{\mathbf{A}\}$, where $\Re\left\{ \cdot\right\}$ and $\Im\left\{ \cdot\right\}$ denote the real and imaginary parts,  respectively. $\mathrm{sign}(\cdot)$ is the sign function applied element-wise to vector or matrix and $\lfloor \cdot \rfloor$ is the floor function. Finally, $j\triangleq \sqrt{-1}$.

\section{Signal Model}
We consider $K$ narrowband far-field signals impinging on a ULA with $M$ elements from different directions $\{\theta_1, \dots, \theta_K\} $. After sampling, the array output can be stacked over the whole $N$ snapshots as
\begin{align}
    \mathbf{X} = \mathbf{A}\mathbf{S}+\mathbf{E},
\end{align}
where $\quad$ $\mathbf{X} = [\mathbf{x}(1), \mathbf{x}(2), \dots, \mathbf{x}(N)] \in \mathbb{C}^{M \times N}$ $\quad$ is the received signal matrix, $\mathbf{A}=\left[\mathbf{a}\left(\theta_1\right), \cdots, \mathbf{a}\left(\theta_K\right)\right] \in \mathbb{C}^{M \times K}$ represents the array steering matrix with $\mathbf{a}(\theta_k)$ denoting the steering vector associated with $k$th source, $\mathbf{S}=[\mathbf{s}(1), \mathbf{s}(2), \dots, \mathbf{s}(N)] \in \mathbb{C}^{ K \times N}$ 
denotes the source signal matrix, and $\mathbf{E}=[\mathbf{e}(1), \mathbf{e}(2), \dots, \mathbf{e}(N)] \in \mathbb{C}^{ M \times N}$ is the noise sequence. The noise has the zero-mean circularly symmetric complex-valued white Gaussian distribution with independent and identically distributed (i.i.d.) known variance $\sigma^2$. Under the case of ULA with half-wavelength antenna spacing , the steering vector $\mathbf{a}(\theta_k)$ can be written as 
\begin{align}
    \vec{a}(\theta_k)&=\begin{bmatrix} 1, e^{j\pi \sin\theta_k},  \dots, e^{j\pi(M-1)\sin\theta_k}\end{bmatrix}^T.
\end{align}
The source signal matrix $\mathbf{S}$ is assumed to be deterministic but unknown, which is referred to as the conditional or deterministic model \cite{60109}.

When one-bit ADC is employed with time-varying threshold for quantization, the array output is modified as
\begin{align}
    \mathbf{Z} = \mathcal{Q}(\mathbf{X} - \mathbf{H}),
\end{align}
where $\mathbf{H} \in \mathbb{C}^{M \times N}$ represents the known threshold and $\mathcal{Q}(\cdot) = \mathrm{sign}(\Re\{\cdot\})+j\mathrm{sign}(\Im\{\cdot\})$ denotes the complex one-bit quantization operator.

We consider a mixed-ADC based architecture  equipped with $M_0$ high-resolution ADCs and  $M_1$  one-bit ADCs, where $M_0 + M_1 = M$. More generally, we define a high-precision ADC indicator vector $\bm{\delta} = [\delta_1, \dots, \delta_M]^T$ with $\delta_i \in \{0, 1\}$, which means that the $i$th antenna is equipped with high-precision ADC when $\delta_i = 1$ and  one-bit ADC when $\delta_i = 0$. So the mixed output can be represented as
\begin{equation}
   \mathbf{Y} = \mathbf{Z} \circ (\bm{\bar{\delta}} \otimes \vec{1}_N^T) + \mathbf{X} \circ (\bm{\delta} \otimes \vec{1}_N^T),
\end{equation}
where  $\bm{\bar{\delta}} = \bold 1_M - \bm{\delta}$ is the indicator for one-bit ADC.


\section{Cramér-Rao Bounds for Mixed Data}

Let $\bm{\varphi}$ collect all the real-valued unknown target parameters, i.e., $\bm{\varphi}=\begin{bmatrix} \bm{\theta}^T,\ \vec{s}_{\rm R}^T, \ \vec{s}_{\rm I}^T \end{bmatrix}^T\in \mathbb{R}^{(K+2KN)\times 1}$, where $ \vec{s} = \rm{vec}(\mathbf{S})$. 

In \cite{shang2022mixed}, it is proved that the Fisher information matrix (FIM) for mixed-ADC data is the summation of the FIMs for the high-precision data and one-bit data, i.e., 
\begin{align}\label{eq:6}
    \mathbf{F}_m(\bm{\varphi}) = \mathbf{F}_0(\bm{\varphi}) + \mathbf{F}_1(\bm{\varphi}),
\end{align}
where $\mathbf{F}_0(\bm{\varphi})$ and $\mathbf{F}_1(\bm{\varphi})$ are FIMs for high-precision data and one-bit data, respectively.

Let the derivatives of the high-precision and one-bit data  with respect to $\bm{\varphi}$ denote as:
\begin{align}
    \mathbf{U}_0 &=\mathbf{U} \text{diag}\left(\mathbf{1}_N \otimes \bm{\delta}\right), \nonumber\\
    \mathbf{U}_1 &=\mathbf{U} \text{diag}\left(\mathbf{1}_N \otimes \overline{\bm{\delta}}\right),
\end{align}
respectively, where 
\begin{align}
    &\mathbf{U}=\left[\mathbf\Delta, \quad\mathbf{G}, \quad j \mathbf{G}\right ]^{H}, \\
    &\mathbf\Delta =\mathbf{S}^T \odot\dot{\mathbf{A}}, \quad \mathbf{G} = \mathbf{I}_N \otimes \mathbf{A}, \\
    &\dot{\mathbf{A}}=\left[\frac{\partial \mathbf{a}\left(\theta_1\right)}{\partial \theta_1}, \ldots, \frac{\partial \mathbf{a}\left(\theta_K\right)}{\partial \theta_K}\right] .
\end{align}
From \cite{stoica2005spectral,li2018bayesian} and (\ref{eq:6}), the FIM for mixed data can be written as
\small
\begin{align}
    \mathbf{F}_{m}(\bm{\varphi}) = &\frac{2}{\sigma^{2}} \Re\left\{\mathbf{U}_{0} \mathbf{U}_{0}^{H}\right\} 
    \nonumber \\ &+
        \frac{1}{\pi \sigma^{2}}\left(\mathbf{U}_{1, \rm R} \boldsymbol{\Lambda}_{\rm R} \mathbf{U}_{1, \rm R}^{T}+\mathbf{U}_{1, \rm I} \boldsymbol{\Lambda}_{\rm I} \mathbf{U}_{1, \rm I}^{T}\right),
\end{align}
\normalsize
where $\mathbf{\Lambda}=\text{diag}([\lambda_1, \ldots, \lambda_{MN}])$ . The diagonal element $\lambda_k$ in $\mathbf{\Lambda}$ is given by
\begin{align}
    \lambda_{k}=B\left(\frac{\Re\left(\zeta_{k}\right)}{\sigma / \sqrt{2}}\right)+j B\left(\frac{\Im\left(\zeta_{k}\right)}{\sigma / \sqrt{2}}\right),
\end{align}
where $\zeta_k$ is the $k$th element in  $\bm{\zeta}=\text{vec}(\mathbf{A}\mathbf{S}-\mathbf{H}) \in \mathbb{C}^{MN \times 1}$ and the function $B(\cdot)$ is defined by
\begin{align}
    B(x)=\left[\frac{1}{\Phi(x)}+\frac{1}{\Phi(-x)}\right] e^{-x^{2}}
\end{align}
with $\Phi(x)=\int_{-\infty}^x \frac{1}{\sqrt{2 \pi}} e^{-\frac{t^2}{2}} d t$ being the cumulative distribution function of the normal standard distribution.

    Considering the optimal time-varying threshold, i.e.,  $\mathbf{H} = \mathbf{A}\mathbf{S}$, the FIM for mixed-ADC data can be simplified as
    \small
    \begin{align}
    \mathbf{F}_{m}(\bm{\varphi}) &= \frac{2}{\sigma^{2}} \Re\left\{\mathbf{U}_{0} \mathbf{U}_{0}^{H}\right\}+\frac{4}{\pi \sigma^{2}}\left(\mathbf{U}_{1, \rm R} \mathbf{U}_{1, \rm R}^{T}+\mathbf{U}_{1, \rm I} \mathbf{U}_{1, \rm I}^{T}\right)\nonumber \\ 
    &=\frac{2}{ \sigma^{2}}(\mathbf{U}_{\rm R}\bar{\mathbf {\Sigma}} \mathbf{U}_{\rm R}^T +\mathbf{U}_{\rm I}\bar{\mathbf {\Sigma}} \mathbf{U}_{\rm I}^T),
    \end{align}
    \normalsize
    where 
    \begin{align}
        \bar{\bm {\Sigma}} = \mathbf{I}_N\otimes \mathbf{\Sigma}_0, \quad
        \bm{\Sigma}_0 = \left(1-\frac{2}{\pi}\right) \text{diag}(\bm{\delta} )+\frac{2}{\pi}\mathbf{I}_M.
    \end{align}
    Let $\bar{\mathbf{U}} = \mathbf{U}\bar{\mathbf {\Sigma}}^{\frac{1}{2}}$,  we have $\bar{\mathbf{U}}=[\bar{\bold\Delta}, \quad \bar{\mathbf{G}}, \quad j \bar{\mathbf{G}}]^H$, where 
    \begin{align}
        \bar{\mathbf{\Delta}} = \bar{\mathbf {\Sigma}}^{\frac{1}{2}}\mathbf{\Delta}
       ,\quad \bar{\mathbf{G}} =\mathbf{I}_N\otimes (\mathbf{\Sigma}_0^{\frac{1}{2}} \mathbf{A}). \label{eq:22}
    \end{align}
    Finally, we can concisely write the FIM for mixed-ADC data as
    \begin{align}
    \mathbf{F}_{m}(\bm{\varphi}) 
    &=\frac{2}{\sigma^2} \Re\left\{\bar{\mathbf{U}} \bar{\mathbf{U}}^{H}\right\}.
    \end{align}
    
    Following the approach given in \cite{stoica2005spectral}, we can obtain  the DOA-related block of the deterministic CRB by block-wise inversion, i.e., 
    \small
    \begin{align} \label{eq:24}
        \vec{CRB}(\mathbf{\theta})=\frac{\sigma^{2}}{2}\Re\left\{\bar{\mathbf\Delta}^{H} \mathbf\Pi_{\bar{\mathbf{G}}}^{\perp} \bar{\mathbf\Delta}\right\}^{-1},
    \end{align}
    \normalsize
    where $\bold\Pi_{\bar{\mathbf{G}}}^{\perp}=\mathbf{I}-\bar{\mathbf{G}}\left(\bar{\mathbf{G}}^H \bar{\mathbf{G}}\right)^{-1} \bar{\mathbf{G}}^H$ is the orthogonal projector onto the null space of $\bar{\mathbf{G}}^H$. Substituting  (\ref{eq:22}) into (\ref{eq:24}), we have a more clear expression for the CRB:
    \small
    \begin{align}\label{eq:25}
     \vec{CRB}(\mathbf{\theta}) 
    =\frac{\sigma^2}{2N}\Re\left\{\left(\dot{\mathbf{A}}^H \mathbf{\Omega} \dot{\mathbf{A}}\right) \circ \hat{\mathbf{P}}^T\right\}^{-1},
    \end{align}
    \normalsize
    where 
    \small
    \begin{align}
        \hat{\mathbf{P}}&=\frac{1}{N} \sum_{t=1}^N \mathbf{s}(t) \mathbf{s}^H(t),\\
        \mathbf{\Omega}
        &=\bold\Sigma_0-\bold\Sigma_0\mathbf{A}(\mathbf{A}^H\bold\Sigma_0\mathbf{A})^{-1}\mathbf{A}^H\bold\Sigma_0.
    \end{align}
 \normalsize

    \subsection{CRB for $K=1, N=1$}
    When the target and snapshot are both single, the  CRB associated with DOA is given as
    \small
    \begin{align} \label{eq:30}
        \mathrm{CRB}(\theta)=\frac{\sigma^2(M_0+\frac{2}{\pi}M_1)}{2pS\pi^2\cos^2\theta } = \frac{M_0+\frac{2}{\pi}M_1}{2\pi^2S} \frac{1}{\text{SNR}\cos^2\theta},
    \end{align}
    \normalsize
    where $p$ is the signal power, $\text{SNR} =\frac{p}{\sigma^2}$ and
    \small
    \begin{align} \label{eq:31}
        S = \sum_{i=1}^Mg_i (i-1)^2\sum_{i=1}^Mg_i-\left[\sum_{i=1}^Mg_i (i-1)\right]^2,
    \end{align}
    \normalsize
    in which $g_i \in\{1, \frac{2}{\pi}\}$, $\sum_{i=1}^M g_i=M_0+\frac{2}{\pi}M_1$, and $M_0+M_1 = M$.
    
    \subsection{Asymptotic CRB}
    For sufficiently large $N$, the estimated power $\hat{\mathbf{P}}$ can be replaced by the true power $\mathbf{P}$. Thus, the CRB is given by
    \label{eq:32}
    \small
    \begin{align}
        \vec{CRB}(\mathbf{\theta})  =\frac{\sigma^2}{2N}\Re\left\{\left(\dot{\mathbf{A}}^H \mathbf{\Omega} \dot{\mathbf{A}}\right) \circ \mathbf{P}^T\right\}^{-1}.
    \end{align}
    \normalsize
    
    When $M$ is sufficiently large, and the number of high-precision ADC is assumed to be constrained (i.e., $M_1$ does not increase proportionally with $M$). We can obtain the following asymptotic result:
    \small
    \begin{align} \label{eq:33}
        \vec{CRB}(\mathbf{\boldsymbol{\theta}}) = \frac{M_0+\frac{2}{\pi}M_1}{2\pi^2 NS}\left[\begin{array}{lll}
            \frac{1}{\text{SNR}_1\cos^2\theta_1} & & 0 \\
            & \ddots & \\
            0 & & \frac{1}{\text{SNR}_K\cos^2\theta_K}
            \end{array}\right],
    \end{align}
    \normalsize
    where $\text{SNR}_i$ is the signal-to-noise ratio for the $i$th signal. This result is motivated by (\ref{eq:30}) and \cite{stoica1989music}, and its detailed derivation can be seen in \cite{zhangmixed}.
    
    It is worth mentioning that the CRB in (\ref{eq:30}) and (\ref{eq:33}) for the mixed-ADC associated with DOA is concerned with the arrangement of the high-precision and one-bit ADCs (i.e., how to maximize $S$), which is rarely mentioned  in the previous. We will study the problem in Section \ref{sec:5}.

\section{Analysis of Arrangement in the Mixed-ADC Based Architecture}\label{sec:5}
Considering the problem mentioned in the last section, the CRB above is achievable using maximum likelihood estimation. So we can get better estimation for DOA if we properly design the arrangement of  high-precision and one-bit ADCs. Intuitively, using Lagrange's identity, we can reformulate the optimization problem as
\small
\begin{align}\label{eq:34} 
    &\max_{\{g_i\}_{i=1,2\cdots,M}} \quad S=\sum_{i=1}^M\sum_{j>i}g_i g_j(j-i)^2 \nonumber \\ 
    & \ \ {\rm s.t.} \quad  g_i \in\{1, \frac{2}{\pi}\}, \quad i=1,2,\dots,M, \nonumber\\
    &  \qquad \quad \  \sum_{i=1}^M g_i=M_0+\frac{2}{\pi}M_1, 
\end{align}
\normalsize
which is called the mixed-precision arrangement problem.
\begin{proposition}\label{pro:1}
The solution to (\ref{eq:34}) is that the high-precision ADCs are placed evenly around the edges of the ULA.
\end{proposition}
\begin{proof}
We adopt the strategy swapping the positions of the high-precision and one-bit ADCs to achieve the optimal arrangement. Generally, let $\tilde{S}(\tilde{g}_m = 1, \tilde{g}_n = \frac{2}{\pi},\dots)$ denote  the arrangement that the positions of the $m$th and $n$th ADCs are swapped based on $S(g_m = \frac{2}{\pi}, g_n = 1,\dots)$ and $M_h = \lfloor \frac{M_0+1}{2} \rfloor$. We consider two special cases. The first case is $m \leq M_h$ and $m$ is the first one-bit ADC's index from left to right. The second is $m \geq M-M_h+1$ and $m$ is the first one-bit ADC's index from right to left. In both cases, $n$ is the index of the first high precision ADC appears from the $m$th to the other side.

Note that $S$ can be written as
\small
\begin{align}
   S=& g_m\sum_{j=1, j\neq m,n}^Mg_j(j-m)^2 + g_n\sum_{j=1, j\neq m,n}^Mg_j(j-n)^2 \nonumber\\
   & + g_mg_n(m-n)^2+C,
\end{align}
\normalsize
where the constant $C$ is not related to $m$ or $n$. Similarly,  we have
\small
\begin{align}
    \tilde{S} - S 
     = \left(1-\frac{2}{\pi}\right)(n-m)\sum_{j=1, j\neq m,n}^Mg_j(2j-m-n).
\end{align}
\normalsize

Let $H(m, n)  \triangleq \sum_{j=1, j\neq m,n}^Mg_j(2j-m-n)$. Note that  $H(m,n)$ has no change if we let $g_m = g_n = 1$. It can be expressed as
\small
\begin{align}
    &H(m, n) = 2\sum_{j=1}^Mg_j j - (m+n)\sum_{j=1}^Mg_j \nonumber\\
    & \geq 2\sum_{j=1}^mj +\frac{4}{\pi}\sum_{j=m+1}^{n-1}j + 2\sum_{j=n}^{n+M_0-m}j \nonumber\\
    &+\frac{4}{\pi}\sum_{j=n+M_0-m+1}^Mj -(m+n)\sum_{j=1}^Mg_j  \triangleq H^{\prime}(m,n).
\end{align}
\normalsize

In the first case, using the facts $ 0<m \leq M_h$ and $m<n \leq M_h+M_1$, we obtain that 
\small
\begin{align}
    &H^{\prime}(m,n) - H^{\prime}(m+1,n) \nonumber\\
    & =(1-\frac{2}{\pi})(2n+3M_0-4m-1)+\frac{2}{\pi}M \nonumber
    \\ &\geq (1-\frac{2}{\pi})(3M_0-2M_h+1)+\frac{2}{\pi}M > 0,
\end{align}
\normalsize
and
\small
\begin{align}
    &H^{\prime}(m,n) - H^{\prime}(m,n+1) \nonumber\\
    &= (1-\frac{2}{\pi})(2m-M_0+1)+\frac{2}{\pi}M
    \nonumber \\ &> \frac{2}{\pi}M-(1-\frac{2}{\pi})M_1> (1-\frac{2}{\pi})M_0 >0.
\end{align}
\normalsize
Combining the above inequality,  we have
\small
\begin{align} \label{eq:67}
    &H(m,n) \geq H^{\prime}(m,n) \geq H^{\prime}(M_h, M_h+M_1) \nonumber \\
    & = M(M_0-2M_h+1)+\frac{2}{\pi}(2M_h+M_1) \nonumber \\
    & = \left\{\begin{array}{rcl}
        (1+\frac{2}{\pi})M & & M_0 \quad \text{is even}\\
        \frac{2}{\pi}(M+1) & & M_0 \quad \text{is odd}
        \end{array} \right. > 0.
\end{align}
\normalsize
So it is proved that $\tilde{S} > S$ in the first case.

In the second case, by letting $j^{\prime} = M+1-j$, we have
\small
\begin{align}
    H(m, n) &= 2\sum_{j=1}^Mg_{j}(M+1-j')-(2M+2-m^{\prime}-n^{\prime})\sum_{j=1}^Mg_j \nonumber\\
    & = (m^{\prime}+n^{\prime})\sum_{j=1}^Mg_j-2\sum_{j=1}^Mg_{j}j^{\prime}.
\end{align}
\normalsize
Due to the fact $m^{\prime} \leq M_h$, it is same as the first case. In summery, we have  $\tilde{S} > S$ in both cases.

Therefore,  we can adjust the positions of high-precision and one-bit ADCs in steps to maximize $S$ until none of the above situations exist. Finally, the high-precision ADCs are distributed evenly around the edges of the ULA like Fig. 1(c), which is the optimal mixed-precision arrangement.
\end{proof}

\begin{remark}
Note that Proposition \ref{pro:1} can be extended to a more complex scenarios that the ULA is equipped with various precision ADCs. Combining the fact that the CRB decreases as the quantization becomes increasingly finer \cite{9664619}, the higher precision ADCs should be placed from center as far as possible to achieve better performance (see more details in \cite{zhangmixed}).
\end{remark}

\section{Simulation and Discussion}
In this section, we present numerical examples to demonstrate the effectiveness of the asymptotic CRB and the optimal mixed-precision arrangement. Assuming the ULA spaced at $d = \frac{1}{2}\lambda$ with $M = 30$ and $M_0 = 10$, we consider three situations showed in  Fig. \ref{fig:arr} for the mixed-ADC based architecture:
\small
\begin{enumerate}
    \item $\{\delta_i = 1\}_{i=1}^{10}$ and $\{\delta_i = 0\}_{i=11}^{30}$ like Fig. 1(a);
    \item $\{\delta_i = 0\}_{i=1}^{10}$, $\{\delta_i = 1\}_{i=11}^{20}$ and $\{\delta_i = 0\}_{i=21}^{30}$ like Fig. 1(b);
    \item $\{\delta_i = 1\}_{i=1}^{5}$, $\{\delta_i = 0\}_{i=6}^{25}$ and $\{\delta_i = 1\}_{i=26}^{30}$, which is the optimal mixed-precision arrangement like Fig. 1(c) that high-precision ADCs are distributed evenly around the edges.
  \end{enumerate}  
  \normalsize
   
    We consider two targets with $\theta_1 = 10^{\circ}, \theta_2 = 20^{\circ}, p_1=p_2 =1$. For the one-bit ADC system and mixed-ADC based architecture, the time-vary threshold has the real and imaginary parts selected randomly and equally likely from a predefined eight-element set $\{-h_{\text{max}},-h_{\text{max}}+\Delta,\dots,h_{\text{max}}-\Delta, h_{\text{max}}\}$ with $h_{\text{max}}=2$ and $\Delta=\frac{h_{\text{max}}}{7}$.
    
    Fig. \ref{fig} shows the CRBs versus $N$ and SNR for $\theta_1$ where situations 1, 2 and 3 are denoted as ``Mixed-ADC1", ``Mixed-ADC2" and ``Mixed-ADC3", respectively. Compared with the one-bit system, the mixed-ADC based architectures can achieve significant performance improvements, especially for the ``Mixed-ADC3''  with large SNR.
    
    When $\text{SNR} = -20$ dB, the optimal threshold can be closed to that by using time-varying threshold scheme. Fig. \ref{fig:1} shows that the CRBs for mixed-ADC data almost coincide with the asymptotic CRB. Also, the optimal mixed-precision arrangement has lower CRB than others. It is shown in Fig. \ref{fig:2} that when the $\text{SNR}$ changes, the optimal mixed-precision arrangement has better performance than 
   others on both asymptotic CRB and actual CRB with time-varying thresholds. In particular, the CRB of the optimal mixed-precision arrangement is almost  $10$ dB lower than others, when the SNR is large. Hence, it is inappropriate to take a random mixed-precision arrangement given the number of high-precision and one-bit ADCs, which may cause a large performance loss especially in the high SNR regime.
    
\begin{figure}[htbp]
\centering
\subfigure[CRB versus $N$, $\text{SNR}=-20$ dB]{
\begin{minipage}[t]{0.5\linewidth}
\centering
\includegraphics[width=1.6in]{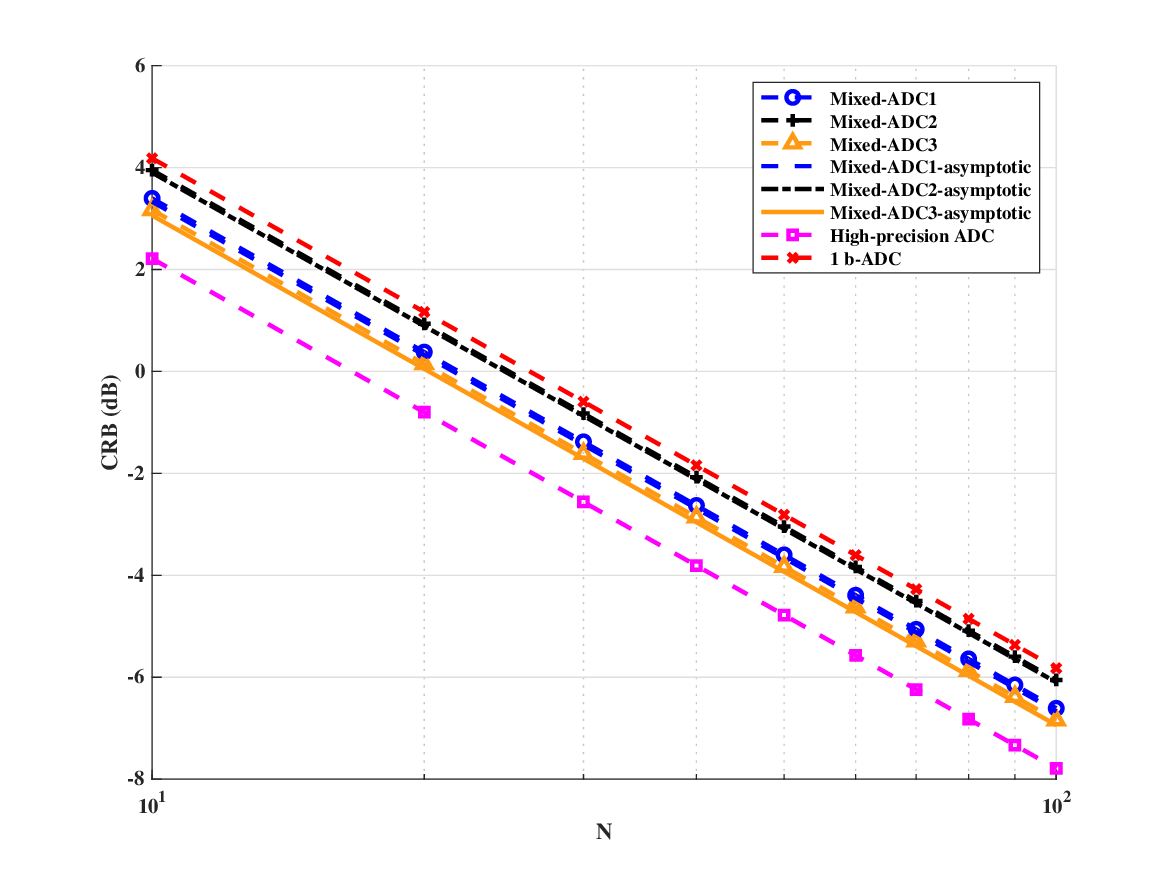}
\label{fig:1}
\end{minipage}%
}%
\subfigure[CRB versus SNR, $N=10$]{
\begin{minipage}[t]{0.5\linewidth}
\centering
\includegraphics[width=1.6in]{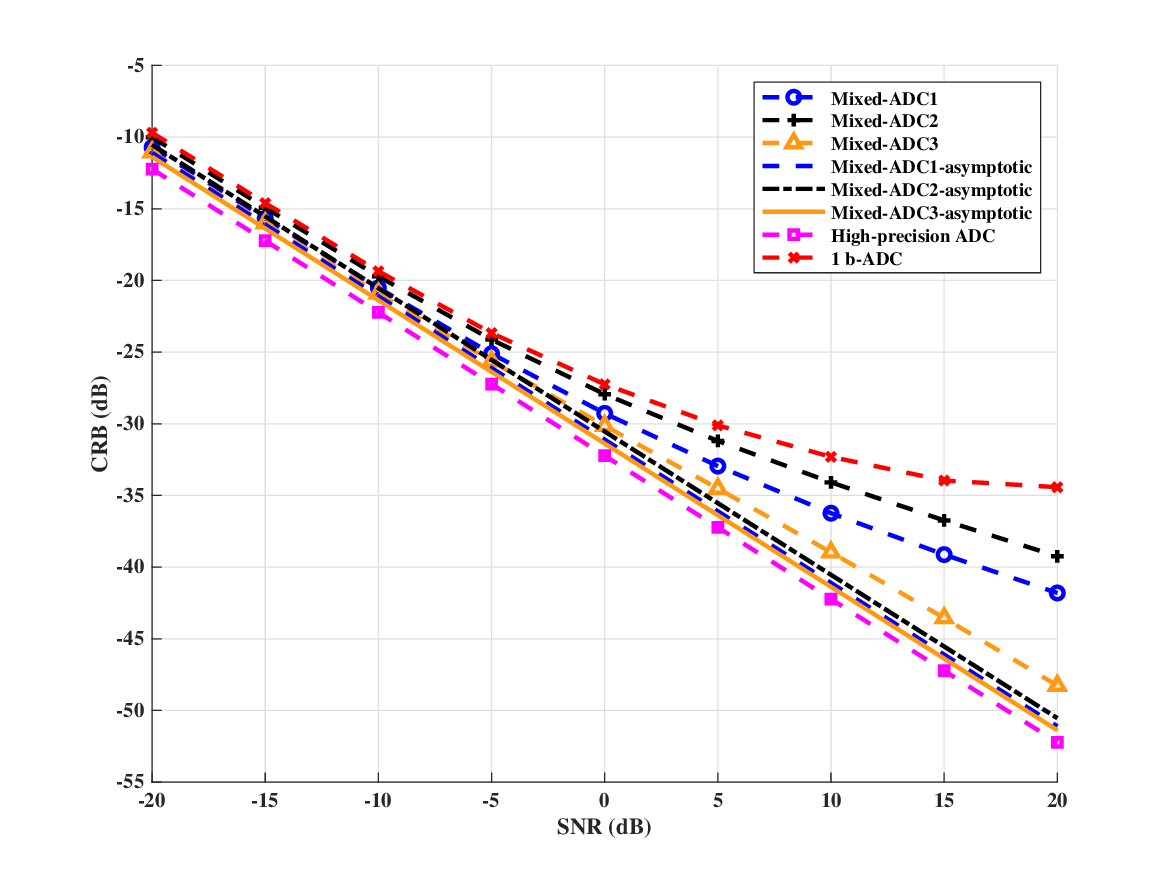}
\label{fig:2}
\end{minipage}%
}%
\centering
\caption{CRB versus $N$ and SNR}
\label{fig}
\end{figure}


\section{Conclusions}
In this work, we have considered the mixed-ADC based architecture for DOA estimation using ULA. We have derived the CRB associated with DOA  under the optimal time-vary threshold. We found the asymptotic CRB is related to the arrangement of high-precision and one-bit ADC. Based on it, we have proved that the high-precision ADCs should be distributed evenly around the edges of ULA to achieve a lower CRB. It can be extended to more general case where the ULA is equipped with various precision ADCs.

\vfill\pagebreak
\bibliographystyle{IEEEbib}
\bibliography{main}

\end{document}